\documentclass[lettersize,journal]{IEEEtran}
\usepackage{graphicx} 
\usepackage{amssymb}
\usepackage{amsmath}
\usepackage{bbm}
\usepackage{tikz}
\usetikzlibrary{automata, positioning}
\usepackage{amsthm}
\usepackage{tikz}
\usetikzlibrary{automata, positioning}
\usepackage{subcaption}
\usepackage{cite}
\usepackage{url}
\usepackage{hyperref}

\newtheorem{theorem}{Theorem}
\newtheorem{lemma}{Lemma}

\begin{document}

\title{Analysis of UAV-Enabled IoT Networks with Energy Harvesting and Wake-Up Radio}

\author{Anthony~Khairallah,~\IEEEmembership{Student Member,~IEEE,}
        Nour~Kouzayha,~\IEEEmembership{Member,~IEEE,}
        Tareq~Y.~Al-Naffouri,~\IEEEmembership{Fellow,~IEEE,}
        and~Hadi~Sarieddeen,~\IEEEmembership{Senior Member,~IEEE.}%
            
\thanks{This work was supported by the AUB University Research Board (URB) and Vertically Integrated Projects (VIP) program and the KAUST Office of Sponsored Research (OSR) under Award No. ORFS-CRG12-2024-6478. Anthony~Khairallah and Hadi~Sarieddeen, are with the Department of Electrical and Computer Engineering (ECE), AUB, Beirut, Lebanon (adk05@mail.aub.edu, hadi.sarieddeen@aub.edu.lb).
Nour~Kouzayha and Tareq Y. Al-Naffouri are with the Department of Computer, Electrical and Mathematical Sciences and Engineering (CEMSE), KAUST, Kingdom of Saudi Arabia (nour.kouzayha@kaust.edu.sa, tareq.alnaffouri@kaust.edu.sa).}
\vspace{-0.9cm}
}


\maketitle

\begin{abstract}
This paper investigates an unmanned aerial vehicle (UAV)-enabled Internet of Things (IoT) architecture that integrates wake-up radio (WuR) and energy harvesting for sustainable device operation. In the proposed system, UAVs transmit radio-frequency (RF) signals that both trigger device activation and replenish stored energy. The IoT device’s behavior is modeled as a discrete-time Markov chain, which captures the evolution of its battery level and operational state (asleep or awake), accounting for both energy harvesting and energy consumption. Leveraging stochastic geometry and discrete-time Markov-chain analysis, we develop a comprehensive mathematical framework to assess system performance. We reveal a tradeoff between transmission frequency and energy consumption, and demonstrate that tuning of system parameters, such as UAV density, can significantly improve both energy efficiency and transmission reliability.

\end{abstract}

\IEEEpeerreviewmaketitle

\vspace{-0.4cm}
\section{Introduction}
\vspace{-0.1cm}

The rapidly expanding Internet of Things (IoT) paradigm envisions massive numbers of connected devices deployed in large-scale and infrastructure-limited environments. These devices often rely on low-power wireless communication and
battery operation, making battery replacement and recharging challenging. Energy-efficient techniques for IoT systems have thus attracted growing interest~\cite{interest2}, particularly energy harvesting (EH) and wake-up radio (WuR) technologies~\cite{Benbuk2023Charging}.

WuR enables ultra-low-power operation by keeping devices in deep sleep and activating them upon receiving a dedicated wake-up signal, thereby drastically reducing idle energy consumption~\cite{RF}. Meanwhile, EH allows IoT devices to convert ambient or dedicated radio frequency (RF) energy into usable power~\cite{energyharvesting1}. While each approach independently offers substantial energy savings, their integration is particularly compelling: WuR minimizes unnecessary wake-ups, and EH replenishes stored energy during sleep periods. Together, they enable near-passive, self-sustaining device operation. Unmanned aerial vehicles (UAVs) further enhance this vision through their mobility and adaptive deployment in challenging IoT environments~\cite{UAV1}. UAVs can act as mobile RF beacons, wireless chargers, and data collectors, making UAV-assisted WuR/EH attractive for mission-critical IoT deployments.

{\color{black}
In this work, we investigate a UAV-assisted IoT architecture combining EH and WuR, where UAVs wake up devices, collect data, and wirelessly recharge batteries. To analyze the coupled interaction between wake-up operation, harvested energy, and uplink transmission, we develop a mathematical framework based on stochastic geometry (SG) and discrete-time Markov chain modeling. Existing works have studied SG-based WuR performance~\cite{SG1,SG2}, and EH-driven battery evolution in terrestrial IoT systems~\cite{QT2,QT1}. Other studies have investigated UAV-assisted wake-up signaling and EH architectures from communication and experimental perspectives, but without jointly modeling battery-state evolution~\cite{SG4,Liu2022UAVWPT,Sheshashayee2022}. The work in~\cite{Ruiz2026Energy} considered WuR-enabled EH management and battery dynamics from an energy-control perspective. 
{\color{black} 
Table~\ref{tab:comparison} summarizes the differences between the proposed framework and prior works. Unlike existing studies, the proposed framework jointly captures UAV-assisted WuR, EH, battery-state evolution, and uplink transmission feasibility. Importantly, wake-up interrupts harvesting, while the battery energy accumulated before activation determines whether an awakened device can transmit. This coupled feedback explains the non-monotonic relation between wake-up and transmission probabilities. Moreover, this joint modeling captures the fundamental trade-offs between wake-up frequency, transmission reliability, and energy efficiency under realistic constraints. The contributions of this paper can be summarized as follows:}


\begin{table}[t]
\centering
{\color{black}
\caption{\color{black}Comparison with representative related works.}
\vspace{-0.1cm}
\footnotesize
\setlength{\tabcolsep}{2.5pt}
\renewcommand{\arraystretch}{0.9}
\begin{tabular}{lccccc}
\hline
\textbf{Work} & \textbf{UAV} & \textbf{WuR} & \textbf{EH} & \textbf{Battery} & \textbf{SG} \\
\hline
\cite{SG2} & $\times$ & \checkmark & $\times$ & $\times$ & \checkmark \\
\cite{QT2} & $\times$ & $\times$ & \checkmark & \checkmark & \checkmark \\
\cite{QT1} & $\times$ & $\times$ & \checkmark & \checkmark & \checkmark \\
\cite{SG4} & \checkmark & \checkmark & $\times$ & $\times$ & \checkmark \\
\cite{Liu2022UAVWPT} & \checkmark & \checkmark & \checkmark & $\times$ & $\times$ \\
\cite{Sheshashayee2022} & \checkmark & \checkmark & $\times$ & $\times$ & $\times$ \\
\cite{Ruiz2026Energy} & $\times$ & \checkmark & \checkmark & \checkmark & $\times$ \\
This work & \checkmark & \checkmark & \checkmark & \checkmark & \checkmark \\
\hline
\label{tab:comparison}
\end{tabular}}
\vspace{-0.9cm}
\end{table}

\begin{enumerate}
\item We develop a unified analytical framework for UAV-enabled IoT networks that jointly models wake-up radio, energy harvesting, and battery state evolution using stochastic geometry and a discrete-time Markov chain.

\item We derive the wake-up and transmission probabilities, and characterize the resulting battery dynamics and energy efficiency, capturing the effects of UAV density, altitude, wake-up threshold, and channel characteristics.

\item  We reveal the coupled trade-off between wake-up frequency, harvested energy accumulation, and uplink transmission feasibility, showing how excessive wake-up activation may degrade transmission performance despite increasing wake-up probability.
\end{enumerate}
}
\vspace{-0.4cm}
\section{System Model and Problem Formulation}
\subsection{Channel and Network Model}
We consider a 2D-network where UAVs follow a homogeneous Poisson point process (PPP), $\Phi = \{\textbf{x}_i\}$ in $\mathbb{R}^2$, with density $\lambda$ (UAVs/m$^2$), and IoT devices are distributed independently as another PPP, $\Phi_{\text{IoT}} = \{\textbf{y}_j\}$ with density $\lambda_{\text{IoT}}$. UAVs hover at a fixed altitude $L_u$ and act as both wake-up beacons and wireless chargers. The considered framework targets hovering or quasi-static UAV deployment regimes commonly adopted during RF charging and data collection phases. {\color{black}Time-varying trajectories, altitude variations, and repositioning would introduce temporal correlation in the serving distance and received RF energy and are left for future work.} We assume that each device is associated with its closest UAV, its serving UAV.
Without loss of generality, we assume a device located at the origin. All UAVs transmit continuously over time slots with a fixed power $P_t$; the received power decays following the distance-dependent pathloss $r_i^{-\eta}$, where $r_i \!=\! \sqrt{||\textbf{x}_i||^2\!+\!L_u^2}$ is the 3D distance from the UAV to the device of interest ($L_u \!\leq\! r_i \!<\! \infty$), and $\eta$ is the path-loss exponent. The environment is also subject to small-scale fading following the Nakagami-$m$ distribution, $h_i \!\sim\! \text{Gamma}(m, m)$, which captures the stronger line-of-sight characteristics associated with UAV-assisted aerial links. {\color{black}The parameter $m$ captures different fading severities, with larger values corresponding to milder LoS-dominated fluctuations.} The instantaneous received RF power in a time slot $t$ is $P_r(t) \!=\! \sum_{\textbf{x}_i\in\Phi} a P_t h_i r_i^{-\eta}$, where $a$ is the power conversion efficiency. 
\vspace{-0.4cm}
\subsection{Wake-Up and Energy Harvesting Protocol}
\vspace{-0.1cm}
We consider a slotted operation in which each slot consists of a downlink (DL) RF beacon phase followed, when applicable, by an uplink (UL) data transmission phase. {\color{black}During the DL phase, the UAV transmits a unified RF beacon that can either activate the WuR when the received power exceeds the sensitivity threshold or be harvested when the device remains asleep.} A device in sleep mode continuously monitors the incoming RF signal using its ultra-low-power WuR receiver, and if the instantaneous received power $P_r(t)$ in that slot exceeds a predefined wake-up threshold $\theta$, the WuR triggers a transition to the active state. {\color{black} This models a low-complexity broadcast WuR based on energy detection, where the received RF envelope is compared locally with a hardware sensitivity threshold and does not require device-specific signaling, UAV-side beam alignment, or instantaneous CSI.} We assume that upon activation, the received RF signal is used to power the WuR receiver, and thus no energy is harvested in that slot; otherwise, if $P_r(t)<\theta$, the device remains asleep and uses the entire RF signal for harvesting.

{\color{black} We consider a slotted system with slot duration $T_s$.} The harvested {\color{black}energy accumulated over $T$ slots} is given as \vspace{-0.15cm}
\begin{equation}\footnotesize
{\color{black}E_S = \sum_{t=t_0}^{t_0+T-1} T_sP_r(t)} 
\leq B,
\label{PST}
\vspace{-0.15cm}
\end{equation}
where $B$ is the battery capacity and $t_0$ is the initial time slot. Note that any energy above the battery level $B$ is discarded. 
\vspace{-0.4cm}
\subsection{Uplink Transmission and Performance Metrics}
\vspace{-0.1cm}
{\color{black} Upon wake-up, the device communicates with its nearest serving UAV and obtains a large-scale path-loss estimate through a pilot or lightweight control signal. This estimate is used for uplink power control; instantaneous small-scale CSI is not required. The device then checks whether its stored energy is sufficient for uplink transmission under channel inversion, which targets a constant received power $\rho$ at the UAV.} The required transmit power at the device is $ \gamma \!=\! \rho r_0^\eta$. If the stored energy exceeds the required transmit power, the device transmits with power $\gamma$ and deducts this amount from the battery, i.e., ${\color{black}E_S(t+1)={\color{black}E_S}(t)-\gamma T_s}$. Otherwise, if ${\color{black}E_S(t)<\gamma T_s}$, the device cannot transmit in that slot and returns to sleep with its battery unchanged, i.e., ${\color{black}E_S}(t+1)={\color{black}E_S}(t)$.

We consider the following performance evaluation metrics:
\begin{itemize}
    \item Wake-up probability, $\mathbb{P}_{\mathrm{w}}$, of the IoT device successfully waking up in a given time slot $t$, i.e., $P_r(t)\ge \theta$.
    \item Transmission probability, $\mathbb{P}_\text{tx}$, of the device waking up successfully and having enough stored energy to transmit at time slot $t$, i.e., $P_r(t) \geq \theta$ and ${\color{black}E_S(t) \geq \gamma T_s}$.
    \item Energy efficiency, $\xi$, {\color{black} defined as the net energy change per successful data transmission (mJ/tx).}
\end{itemize}
\vspace{-0.4cm}
\section{Performance Analysis}
\subsection{Wake-up Probability}
The wake-up probability, defined as the probability that the instantaneous received power exceeds the wake-up threshold $\theta$, is obtained as the complementary cumulative distribution function (CCDF) of $P_r(t)$, 
\begin{equation}\footnotesize
    \mathbb{P}_\text{w} = 1 - F_{P_r}(\theta),
    \label{eq:wakeup}
\end{equation}
where $F_{P_r}(\theta)$ is the CDF of the received power. {\color{black}In the low-density regime, the received power can be approximated by the dominant UAV approximation, where only the serving UAV is considered explicitly, and the contribution of all other UAVs is replaced by their average power~\cite{SG1,domUAV2}},
\begin{equation}\footnotesize
\label{Pr_approx}
P_r(t) \approx a P_t h_0 r_0^{-\eta} + \psi(r_0),\end{equation}
where $h_0$ and $r_0$ are the fading gain and the distance to the serving UAV, and 
\begin{equation}\footnotesize
    \psi(r_0) = \mathbb{E}\Big[\sum_{\textbf{x}_i\in\Phi\setminus\{\textbf{x}_0\}} a P_t h_i r_i^{-\eta}\Big]=\frac{2\pi\lambda a P_t}{\eta - 2}r_0^{2-\eta} 
    \label{eq:psi}
\end{equation}
is obtained via Campbell's theorem, assuming $\mathbb{E}[h_i]=1$~\cite{SG1}. The wake-up probability is derived in Theorem~\ref{theorem:PrCDF}. 
\begin{theorem}
    Using the dominant UAV approximation, the CDF of the received power in a given time slot is given as
\begin{equation}\footnotesize
    F_{P_r}(s) = \int_u^\infty\!\! \frac{1}{\Gamma(m)}\gamma\left[m , m\Big(\frac{s r_0^\eta}{a P_t} - \frac{2\pi\lambda r_0^2}{\eta - 2}\Big)\!\right]\!f_R(r_0) dr_0,
    \label{eq:CDFPr}
\end{equation}
where $u = \max\left\{L_u, \Big(\frac{2\pi\lambda a P_t}{(\eta-2)s} \Big)^\frac{1}{\eta-2}\right\}$, $\Gamma(\cdot)$ and $\gamma[\cdot,\cdot]$ are the Gamma and the lower incomplete gamma functions, and $f_R(r_0)=2\pi \lambda r e^{-\pi\lambda(r^2-L_u^2)}$ is the PDF of the UAV-device 3D distance obtained from the void probability of a PPP. The wake-up probability is obtained by plugging~\eqref{eq:CDFPr} in~\eqref{eq:wakeup} ($s=\theta$).
\begin{IEEEproof}
    See Appendix A.
\end{IEEEproof}
\label{theorem:PrCDF}
\end{theorem}
\vspace{-0.7cm}

\subsection{Battery State Evolution Analysis}

We discretize the battery capacity $B$ into $L=\frac{B}{w}$ equally spaced levels of size $w$ and model the system using a two-dimensional discrete-time Markov chain (DTMC). Each state is denoted by $(l,b)$, where $l\!\in\!\{0,1,\ldots,L\}$ represents the battery level in discrete units of $w$ joules, and $b\!\in\!\{0,1\}$ indicates the device mode ($b=0$ for sleep mode and $b=1$ for active mode). The wake-up threshold is set to $\theta=\tau w/T_{s}$, with $\tau\in\mathbb{N}$. We define $p_i$ and $q_i$ as the probabilities that the harvested energy in a time slot lies between $iw$ and $(i+1)w$, and that the required transmit energy lies between $(i-1)w$ and $iw$, respectively. {\color{black}
For tractability, spatial randomness is averaged in the DTMC by using the unconditional distributions $F_{P_r}$ and $F_\gamma$, so the correlation induced by $r_0$ is not tracked. The resulting DTMC characterizes the average battery dynamics of a typical device over the PPP geometry, rather than the conditional battery evolution for a fixed serving distance.} Hence, $p_i$ and $q_i$ can be expressed as
    \begin{equation}\footnotesize
        p_i = 
         F_{P_r}((i+1)w) - F_{P_r}(iw), \quad  i \!=\! 0, 1, ..., \tau-1, 
    \end{equation}
     and \vspace{-10pt}    \begin{equation}\footnotesize
q_i = F_{\gamma}(iw/T_s) - F_{\gamma}((i-1)w/T_s), \quad i = 1,2,\ldots,L,        \label{eq:qi}
    \end{equation}
where $F_{P_r}(\cdot)$ is given in~\eqref{eq:CDFPr} and $F_{\gamma}(\cdot)$ is the CDF of the required transmit power and is given in Lemma~\ref{lemma:transmit}.
\begin{lemma}
The CDF of the power required to transmit data from the typical device implementing power control to the serving UAV upon successful wake-up is given as~\cite{QT2}, {\color{black}with the minimum transmit power $\rho L_u^{\eta}$ imposed by the UAV altitude.
\begin{equation}\footnotesize
F_\gamma(s) =
\begin{cases}
0, & s < \rho L_u^{\eta}, \\
1 - \exp\left[-\pi\lambda\left(\left(\frac{s}{\rho}\right)^\frac{2}{\eta} - L_u^2\right)\right], & s \ge \rho L_u^{\eta}.
\end{cases}
\vspace{-0.5cm}
\end{equation}\label{lemma:transmit}}
\end{lemma}
The probabilities $\{p_i\}$ and $\{q_i\}$ describe upward and downward transitions due to harvesting while asleep, and energy consumption when active. To ensure consistency,
\begin{equation}\footnotesize
\label{unity_sum}
\sum_{i=0}^{\tau-1}p_i + \mathbb{P}_\text{w} =1,
\end{equation}
where $\mathbb{P}_w$ denotes the probability that the harvested power exceeds the wake-up threshold, causing the device to stop energy accumulation and switch to the active state. The state transitions are classified into four cases, for $i, j \!\in\! \{0, 1, ..., L\}$:
\begin{enumerate}
    \item Awake $(i,1)\rightarrow$ Awake $(j,1)$:
The device cannot remain awake across consecutive slots, as transmission is attempted immediately and the device then returns to sleep. Hence, the transition probability is
        \begin{equation}\footnotesize
    \textbf{P}_{(i,1), (j, 1)} = 0.
    \end{equation}
    
    \item Awake $(i,1)\rightarrow$ Asleep $(j,0)$:
After successful activation, the device attempts transmission and then returns to sleep. If $j<i$, $(i-j)w$ energy units are consumed with probability $q_{i-j}$. If transmission does not occur due to insufficient stored energy, the battery level remains unchanged ($j=i$) with probability $1-\sum_{k=1}^i q_k$. Thus, the transition probability is
    \begin{equation}\footnotesize
    \textbf{P}_{(i, 1), (j, 0)} = \begin{cases}
        q_{i-j} & j < i, \\
        1 - \sum_{k=1}^i q_k & j = i,\\
        0, & \text{otherwise}.
    \end{cases}
    \end{equation}

    \item Asleep $(i,0)\rightarrow$ Asleep $(j,0)$:
In this case, the device harvests energy but remains below the wake-up threshold $\tau$. The battery level increases by $(j-i)w$ with probability $p_{j-i}$ for $j-i<\tau$. If the battery saturates at $j=L$, any excess harvested energy is discarded. The transition probability is given by
    \begin{equation}\footnotesize
    \textbf{P}_{(i, 0), (j, 0)} = \begin{cases}
        p_{j-i}, & i \leq j  \text{ and $j-i < \tau$,} \\
        \sum_{k=L-i}^{\tau-1}p_{k}, & i \leq j = L, \\
        0, & \text{otherwise}.
    \end{cases}
    \end{equation}

    \item Asleep $(i,0)\rightarrow$ Awake $(j,1)$:
Represents a wake-up event when the received power exceeds $\theta$. The battery level remains unchanged. The transition probability is
    \begin{equation}\footnotesize
        \textbf{P}_{(i, 0), (j, 1)} = \begin{cases}
    \mathbb{P}_\text{w}, & i = j, \\
    0 & \text{otherwise}.
    \end{cases}
    \end{equation}
\end{enumerate}

For notation simplicity, we redefine $p'_i$ to account for the case when harvesting remains below threshold and the device stays asleep. Thus, for $\sum_{i=0}^{L-1}p'_i =1- \mathbb{P}_\text{w}$, $p'_i$ is expressed as
\begin{equation}\footnotesize
p'_i = p_i \mathbbm{1}(i < \tau) = \begin{cases}
        p_i, & i < \tau, \\
        0, & \text{otherwise}.
    \end{cases}
    \label{eq:ppi}
    \end{equation}

The $2(L+1)\times 2(L+1)$ DTMC transition matrix $\mathbf{P}$ is defined over the state space $\mathcal{S}=\{(0,1),(0,0),(1,1),(1,0),\dots,(L,1),(L,0)\}$, as illustrated in Theorem 2. Fig.~\ref{fig:state_fig} shows a small example with battery capacity $B=3w$ ($L=3$) and wake-up threshold $\theta=2w$ ($\tau=2$), where possible harvesting increments are $w$, $2w$, $3w$. The transitions highlight how the device alternates between EH in the sleep state and transmitting in the active state.

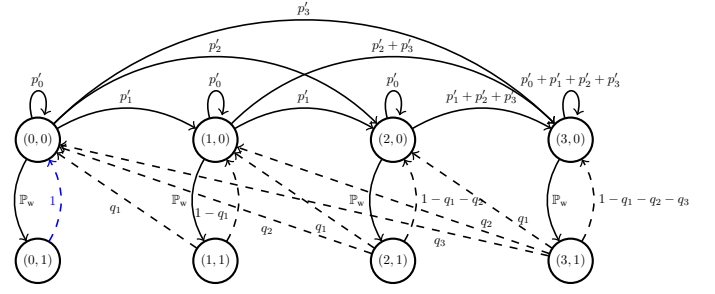
\begin{figure}[t!]
    \centering
    \begin{tikzpicture}[->, auto, node distance=2cm and 3.5cm, scale=0.5, transform shape, semithick]
        \tikzstyle{every state}=[draw=black, thick, minimum size=1cm]

        \node[state] (A) {$(0,0)$};
        \node[state] (B) [right=of A] {$(1,0)$};
        \node[state] (C) [right=of B] {$(2,0)$};
        \node[state] (D) [right=of C] {$(3,0)$};
        \node[state] (E) [below=of A] {$(0,1)$};
        \node[state] (F) [below=of B] {$(1,1)$};
        \node[state] (G) [below=of C] {$(2,1)$};
        \node[state] (H) [below=of D] {$(3,1)$};

        \path 
            (E) edge[bend right=30,dashed,color=blue] node {\(1\)} (A)
            (A) edge[loop above] node {\(p'_0\)} (A)
                edge[bend left=30] node {\(p'_1\)} (B)
                edge[bend left=45] node {\(p'_2\)} (C)
                edge[bend left=45] node {\(p'_3\)} (D)
                edge[bend right=30] node {\(\mathbb{P}_\text{w}\)} (E)
            (F) edge[bend right=30,dashed,color=black] node[pos=0.2] {\(1-q_1\)} (B)
            (B) edge[loop above] node {\(p'_0\)} (B)
                edge[bend left=30] node {\(p'_1\)} (C)
                edge[bend left=45] node {\(p'_2+p'_3\)} (D)
                edge[bend right=30] node[left] {\(\mathbb{P}_\text{w}\)} (F)
            (G) edge[bend right=30,dashed,color=black] node[right] {\(1-q_1-q_2\)} (C)
            (C) edge[loop above] node {\(p'_0\)} (C)
                edge[bend right=30] node[left] {\(\mathbb{P}_\text{w}\)} (G)
                edge[bend left=30] node {\(p'_1+p'_2+p'_3\)} (D)
            (H) edge[bend right=30,dashed,color=black] node[right] {\(1-q_1-q_2-q_3\)} (D)
            (D) edge[loop above] node {\(p'_0+p'_1+p'_2+p'_3\)} (C)
                edge[bend right=30] node {\(\mathbb{P}_\text{w}\)} (H)
            (F) edge[dashed] node {\(q_1\)} (A)
            (G) edge[dashed] node[pos=0.3] {\(q_2\)} (A)
                edge[dashed] node[pos=0.3] {\(q_1\)} (B)
            (H) edge[dashed] node[pos=0.2] {\(q_3\)} (A)
                edge[dashed] node[pos=0.2,above] {\(q_2\)} (B)
                edge[dashed] node[pos=0.2,above] {\(q_1\)} (C);

    \end{tikzpicture}
    \caption{Markov chain with $L=3$ and $\tau=2$.}
    \label{fig:state_fig}
    \vspace{-0.6cm}
\end{figure}
\begin{theorem}
        The transition matrix that characterizes the joint dynamics of the device's battery and operational states in a UAV-enabled IoT system, with WuR and EH is given as 
        \begin{equation}\footnotesize
            \textbf{P} = \begin{pmatrix}
    0 & 1 & 0 & 0 & \hdots & 0 & 0 \\
    \mathbb{P}_\text{w} & p'_0 & 0 & p'_1 & \hdots & 0 & p'_L \\
    0 & q_1 & 0 & 1 - q_1 & \hdots & 0 & 0 \\
    0 & 0 & \mathbb{P}_\text{w} & p'_0 & \hdots & 0 & p'_{L-1} + p'_L\\
    \vdots & \vdots &  \vdots &\vdots& \ddots & \vdots & \vdots\\
    0 & q_L & 0 & q_{L-1} &  \hdots & 0 & 1-q_1-...-q_L \\
    0 & 0 & 0 & 0 &  \hdots & \mathbb{P}_\text{w} & p'_0 + ... + p'_L \\
\end{pmatrix},
        \end{equation}
    where $p'_i$ and $q_i$ are given in \eqref{eq:ppi} and \eqref{eq:qi} respectively. 
    \end{theorem}


\vspace{-0.5cm}
\subsection{Transmission Probability}
With the DTMC and its transition matrix defined, the steady-state probability vector $\boldsymbol{\pi}$ is obtained by solving
\begin{equation}\footnotesize
    \begin{cases}
        \boldsymbol{\pi}\textbf{P} = \boldsymbol{\pi} \\
        \boldsymbol{\pi 1} = 1,
    \end{cases}
\end{equation}
where $\boldsymbol{1} = \begin{pmatrix}
    1 & 1 & \hdots & 1
\end{pmatrix}^t$.
To reduce computational complexity, we use the Sheskin algorithm~\cite{Sheskin}, which avoids direct matrix inversion by iteratively reducing the transition matrix using additions and multiplications. Once reduced to a single equation, the steady-state probabilities are computed via back-substitution. This approach is particularly efficient for large, structured Markov chains. After obtaining the steady-state probabilities $\boldsymbol{\pi}$, the transmission probability $\mathbb{P}_\text{tx}$ is calculated as the average probability of transmitting over all awake states,
\begin{equation}\footnotesize
\mathbb{P}_\text{tx} = \sum_{i=1}^L\boldsymbol{\pi}_{(i, 1)}\sum_{j=1}^iq_j.
\label{eq:Ptx}
\end{equation}
\vspace{-0.6cm}
\section{Energy Efficiency Analysis}
We evaluate the energy efficiency by analyzing a single operational cycle, defined as the period starting when the device is in the sleep mode harvesting RF energy, followed by a wake-up event, an attempted transmission (if feasible), and a return to sleep. Each cycle consists of a random number of harvesting slots, exactly one wake-up slot, and at most one transmission slot. {\color{black}The following analysis evaluates the average net energy variation per successful transmission under the SG/DTMC framework developed above. Accordingly, the resulting metric inherits the previously stated assumptions on the dominant-UAV received-power approximation, average spatial geometry, quasi-static UAV service phase, RF harvesting model, and uplink energy consumption.} We denote the received energy per slot by $E_r(t)=T_sP_r(t)$ and assume $T_s=1$.

We define $E_G$ as the harvested energy during this cycle, $E_T$ as the energy consumed for transmission, and $E_S(t_0)$ as the stored battery energy at the beginning of the cycle. The net energy variation per cycle is expressed as $\Delta E = E_G - E_T$.
The stored energy at the end of the cycle is given as $E_S(t_0+\Delta T) = E_S(t_0) + \Delta E$, where $\Delta T$ is the cycle duration.
The expected battery energy at the beginning of the cycle is {\color{black}
\begin{equation}\footnotesize
\mathbb{E}[E_S(t_0)] =
\sum_{j = 0}^{L} j w
\sum_{i = 0}^{L} \pi_{i \mid \text{awake}} \mathbf{P}_{(i,1),(j,0)},
\label{ENERGY1}
\end{equation}
where $\pi_{i \mid \text{awake}} =
\frac{\boldsymbol{\pi}_{(i,1)}}{\sum_{\ell=0}^{L}\boldsymbol{\pi}_{(\ell,1)}}$ is the steady-state probability that the device is in battery level $i$ conditioned on being awake.}

The harvested energy is calculated as the accumulated energy until the first wake-up event occurs, i.e., until $P_r(t)$ exceeds the wake-up threshold $\theta$. Let $H$ denote the number of harvesting slots; the harvested energy in one cycle is
\begin{equation}\footnotesize
E_G = \min\left\{\sum_{t = t_0}^{H} E_r(t), \quad B  - E_S(t_0)\right\}.
\end{equation}
The number of harvesting slots $H$ follows a geometric distribution with success probability $\mathbb{P}_w$, giving an average of $(1-\mathbb{P}_w)/\mathbb{P}_w$. Each such slot contributes an expected harvested energy of $\mathbb{E}[E_r | P_r< \theta]$. Combining these facts and using the concavity of the $\min$ operator yields
\begin{equation}\footnotesize
\mathbb{E}[E_G] \leq \min\left\{\frac{1-\mathbb{P}_\text{w}}{\mathbb{P}_\text{w}}T_s\mathbb{E}[P_r | P_r < \theta],\quad B - \mathbb{E}\left[E_S(t_0)\right] \right\},
\label{ENERGY2}
\end{equation}
where the conditional term $\mathbb{E}[P_r | P_r < \theta]$ is derived from the CDF of the received power and is given as
\begin{equation}\footnotesize
    \mathbb{E}[P_r | P_r < \theta] = \theta-\frac{1}{F_{P_r}(\theta)}\int_{0}^{\theta} F_{P_r}(x)dx.
    \label{eq:cond}
\end{equation}
\begin{figure*}[h]
    \centering
    \begin{subfigure}[t]{0.32\textwidth}
        \includegraphics[width=\linewidth]{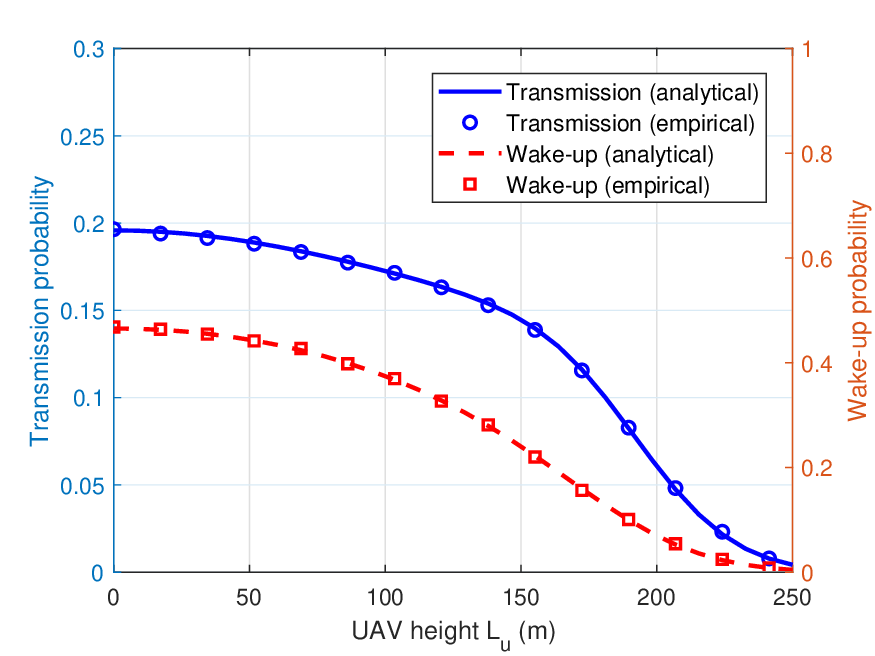}
        \caption{UAV altitude}
        \label{fig:Lu}
    \end{subfigure}
    \hfill
    \begin{subfigure}[t]{0.32\textwidth}
        \includegraphics[width=\linewidth]{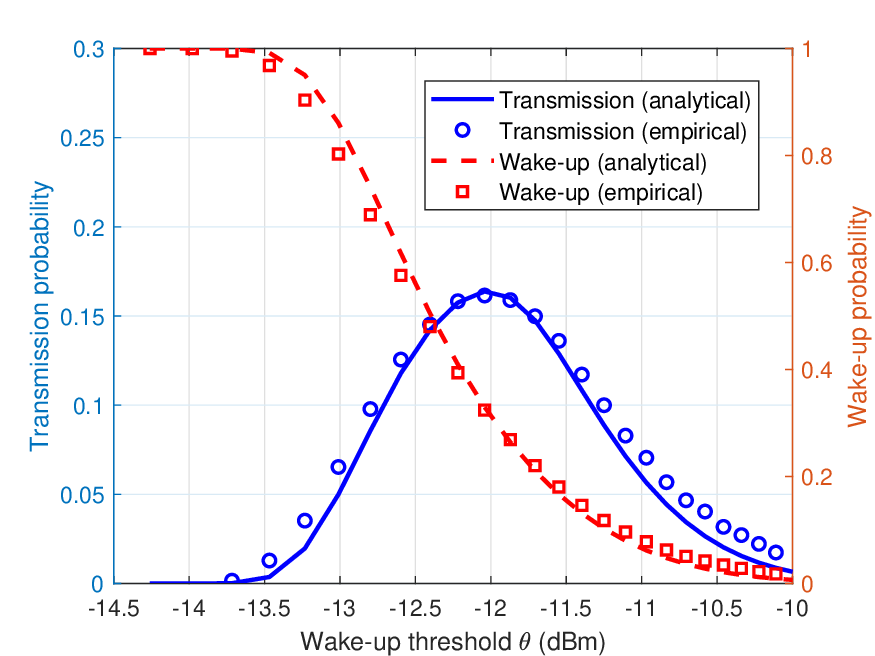}
        \caption{Wake-up threshold}
        \label{fig:theta}
    \end{subfigure}
    \hfill
    \begin{subfigure}[t]{0.32\textwidth}
        \includegraphics[width=\linewidth]{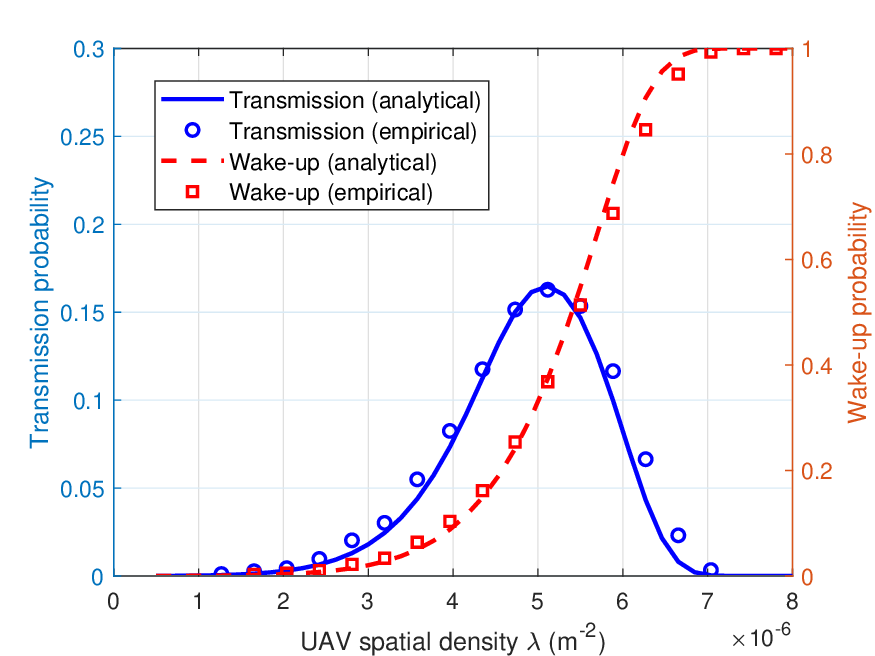}
        \caption{UAV density}
        \label{fig:lambda}
    \end{subfigure}
    \caption{Impact of system parameters on transmission and wake-up probabilities.}
    \label{fig:combined}
    \vspace{-0.75cm}
\end{figure*}
Finally, the transmission energy is obtained by averaging over all active states and is given as
{\color{black}
\begin{equation}\footnotesize
\mathbb{E}[E_T] =
\sum_{i=0}^{L} \pi_{i \mid \text{awake}}
\sum_{j=1}^{i} q_j j w.
\label{ENERGY4}
\end{equation}
}

The expected number of cycles required before a successful transmission is denoted by $N$. {\color{black}For tractability, we approximate successive operational cycles as independent transmission attempts with success probability $\mathbb{P}_\text{tx}$, where $\mathbb{P}_\text{tx}$ is given in~\eqref{eq:Ptx}. Under this approximation, $N$ follows a geometric distribution, yielding $\mathbb{E}[N] \approx 1/\mathbb{P}_\text{tx}$.} Consequently, the net energy variation per successful transmission, reflecting the energy efficiency of the proposed architecture, is approximated as 
\vspace{-0.1cm}
\begin{equation}\footnotesize
\label{PTX}
\xi = \mathbb{E}[N]\mathbb{E}[\Delta E] \approx \frac{\mathbb{E}[\Delta E]}{\mathbb{P}_\text{tx}}.
\end{equation}

\vspace{-0.6cm}
\section{Numerical Results}
In this section, we validate the analytical framework using Monte Carlo simulations and illustrate the impact of key system parameters. Unless otherwise stated, default values are \( \lambda = 5 \times 10^{-6} \, \text{m}^{-2} \), \( a = 1 \), \( L_u = 120 \, \text{m} \), \( P_t = 1 \, \text{W} \), \( \eta = 2.2 \), \( m = 4 \), \( \rho = 1 \times 10^{-9} \, \text{W} \), \( \theta = -12 \, \text{dBm} \), and \( B = 1 \, \text{mJ} \).

{\color{black} The results in Fig.~\ref{fig:combined} indicate close agreement between analytical and empirical wake-up and transmission probabilities across different parameter sweeps, confirming the validity of the DTMC approximation and the accuracy of the UAV approximation over the considered altitude, threshold, and density ranges.} Fig.~\ref{fig:Lu} shows the wake-up probability $\mathbb{P}_\text{w}$ and transmission probability $\mathbb{P}_\text{tx}$ as functions of the UAV altitude $L_u$. As $L_u$ increases, the wake-up probability decreases monotonically due to increased path loss. {\color{black} At the same time, the smaller received power also limits the amount of harvested energy, making successful uplink transmissions less likely and therefore reducing the transmission probability.}

Fig.~\ref{fig:theta} illustrates the impact of the wake-up threshold $\theta$. Since $\theta$ reflects the hardware sensitivity of the WuR circuit, a higher threshold reduces the likelihood of activation and, as a result, the wake-up probability. The transmission probability, however, is close to zero for low values of $\theta$, as devices wake up almost in every slot without being able to accumulate sufficient energy in the battery. At moderate thresholds, wake-ups become less frequent, allowing devices to accumulate energy, and the transmission probability increases sharply. At very high thresholds, the transmission probability drops gradually as wake-up events become rare. 

Fig.~\ref{fig:lambda} illustrates the impact of UAV density $\lambda$. As $\lambda$ increases, the wake-up probability rises and quickly saturates, as the received power increases with more UAVs. The transmission probability initially increases because wake-up events are infrequent and harvesting intervals are long. As the wake-up probability saturates with additional UAVs, devices wake up almost continuously with insufficient stored energy, leading to a sharp decline in transmission probability. This highlights the importance of jointly modeling WuR and EH, which is uniquely enabled by the proposed Markov chain framework.


Fig.~4 examines the cycle-level energy terms versus the wake-up threshold $\theta$. Increasing $\theta$ makes wake-up events less frequent, allowing longer harvesting periods and larger stored energy. Hence, $E_G$ initially increases with $\theta$, but eventually peaks when wake-ups become too rare. The transmission energy $E_T$ also increases because higher thresholds tend to activate devices after more energy has been accumulated. The net variation $\Delta E=E_G-E_T$, therefore, captures the balance between harvesting and transmission expenditure and exhibits a non-monotonic trend. Its peak occurs at a different threshold from the one maximizing $P_{\rm tx}$ in Fig.~3(b), showing that maximum transmission probability and maximum net energy surplus are not achieved at the same operating point. Positive $\Delta E$ values indicate regimes where the device is energy-sustainable on average under the proposed model.

\begin{figure}[t!]
    \centering
    \includegraphics[width=0.7\linewidth]{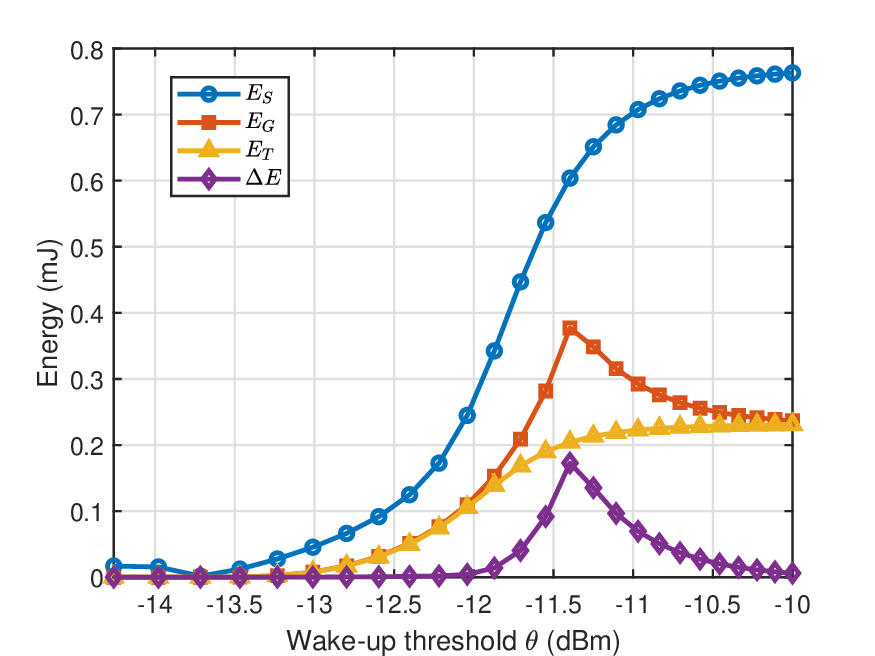}
    \caption{\color{black}Cycle-level energy terms versus wake-up threshold $\theta$.}
    \label{energy}
    \vspace{-0.7cm}
\end{figure}

\vspace{-0.35cm}
\section{Conclusion}
\vspace{-0.1cm}
This paper introduced an analytical framework for UAV-assisted IoT networks that combine energy harvesting and wake-up radio technologies to improve energy efficiency and device sustainability. Modeling device operation with a discrete-time Markov chain and applying stochastic geometry and queuing theory, we characterized key behaviors including energy harvesting, wake-up, and data transmission. The results highlight the advantages of integrating energy harvesting with wake-up signaling, enabling more autonomous and durable IoT deployments, especially in remote areas. 

\vspace{-0.4cm}
\section*{Appendix A}
\vspace{-0.2cm}
{\color{black}To obtain the received-power CDF, we condition on the serving-UAV distance $r_0$. Under the dominant-UAV approximation, the serving UAV contributes the random fading-dependent term, whereas the aggregate contribution of the remaining UAVs is represented by its conditional mean $\psi(r_0)$.} The CDF of the received power $P_r$ can be derived as
\begin{equation}\footnotesize
 \begin{aligned}
    &F_{P_r}(s) 
    = \mathbb{P}\left[a P_t h_0 r_0^{-\eta} + \psi(r_0) \leq s\right] \\
    &= \mathbb{E}_{R}\left[\mathbb{P}\left[h_0 \leq \frac{s - \psi(r_0)}{aP_t}r_0^\eta\right] \Big| R = r_0\right] \\
    &\overset{(a)}{=} \mathbb{E}_{R}\left[ 
    F_{h_0} \left( \frac{s - \psi(r_0)}{aP_t}r_0^\eta\right)\mathbbm{1}\left( \frac{s - \psi(r_0)}{aP_t}r_0^\eta > 0 \right) \Big| R = r_0\right] \\
    &\overset{(b)}{=} \!\!\!\int_{L_u}^{\infty}\!\!\! F_{h_0} \!\left(\! \frac{s r_0^\eta}{aP_t} \!-\! \frac{2\pi\lambda r_0^2}{\eta - 2}\right)\! \mathbbm{1}\!\left( r_0^{\eta-2} \!>\!\! \frac{2\pi\lambda aP_t}{(\eta-2)s} \right)\! f_R(r_0) \! \, dr_0, \\
\end{aligned}   
\end{equation}
where (a) follows from conditioning on the main-link distance $r_0$, and (b) follows from evaluating the expectation and substituting $\psi(r_0)$ using \eqref{eq:psi}. The expression in \eqref{eq:CDFPr} is obtained by applying the gamma CDF $F_{h_0}(\cdot)$ corresponding to Nakagami-$m$ fading and by denoting $u = \max\{L_u, \left(\frac{2\pi\lambda aP_t}{(\eta-2)s} \right)^\frac{1}{\eta-2} \}$.

\vspace{-0.5cm}
\bibliographystyle{IEEEtran}
\bibliography{references}

\end{document}